\documentclass[letterpaper,11pt]{article}

\usepackage[margin=1in]{geometry}  
\usepackage{xspace,exscale,relsize}
\usepackage{fancybox,shadow}
\usepackage{graphicx}
\usepackage{color}
\usepackage{amsthm,amsfonts}
\usepackage{amssymb}
\usepackage{authblk}
\usepackage[title]{appendix}

\usepackage{makecell}

\usepackage{microtype}

\usepackage[
            CJKbookmarks=true,
            bookmarksnumbered=true,
            bookmarksopen=true,
            colorlinks=true,
            citecolor=red,
            linkcolor=blue,
            anchorcolor=red,
            urlcolor=blue,
            pdfauthor={Polyanskiy and Wu}
            ]{hyperref}

\usepackage[noadjust]{cite}

\usepackage{amsmath}
	\makeatletter
	\let\over=\@@over \let\overwithdelims=\@@overwithdelims
	\let\atop=\@@atop \let\atopwithdelims=\@@atopwithdelims
  	\let\above=\@@above \let\abovewithdelims=\@@abovewithdelims
  	\makeatother
\usepackage{fixltx2e}

\usepackage{rotating}

\usepackage{ifpdf}

\usepackage{subfigure}
\usepackage{psfrag}

\usepackage{prettyref,enumerate}

\usepackage[all]{xy}

\usepackage{tikz}
\usetikzlibrary{arrows}
\tikzstyle{int}=[draw, fill=blue!20, minimum size=2em]
\tikzstyle{dot}=[circle, draw, fill=blue!20, minimum size=2em]
\tikzstyle{init} = [pin edge={to-,thin,black}]

\usepackage{tikz-cd}

\ifx\eqref\undefined
	\newcommand{\eqref}[1]{~(\ref{#1})}
\fi
\ifx\mod\undefined
	\def\mod{\mathop{\rm mod}}
\fi

\def\EE{\Expect}

\def\PP{\mathbb{P}}

\def\eqdef{\triangleq}

\def\perc{\mathrm{perc}}
\def\ER{\mathrm{ER}}

\newcommand{\ERG}{Erd\"os-R\'enyi\xspace}

\newcommand{\stepa}[1]{\overset{\rm (a)}{#1}}
\newcommand{\stepb}[1]{\overset{\rm (b)}{#1}}
\newcommand{\stepc}[1]{\overset{\rm (c)}{#1}}
\newcommand{\stepd}[1]{\overset{\rm (d)}{#1}}

\newcommand{\BSC}{\mathsf{BSC}}
\newcommand{\BEC}{\mathsf{BEC}}

\newcommand{\Binom}{\mathrm{Binom}}
\newcommand{\Unif}{\mathrm{Unif}}

\newcommand{\LC}{\mathrm{LC}}

\newcommand{\etaKL}{\eta_{\rm KL}}

\newcommand{\reals}{\mathbb{R}}

\newcommand{\integers}{\mathbb{Z}}
\newcommand{\Expect}{\mathbb{E}}
\newcommand{\expect}[1]{\mathbb{E}\left[#1\right]}
\newcommand{\eexpect}[1]{\mathbb{E}[#1]}

\newcommand{\prob}[1]{\mathbb{P}\left[#1\right]}
\newcommand{\probs}[2]{\mathbb{P}_{#2}\left[#1\right]}
\newcommand{\pprob}[1]{\mathbb{P}[#1]}

\newcommand{\TV}{d_{\rm TV}}

\newcommand{\iid}{i.i.d.\xspace}

\newcommand{\lnorm}[2]{\left\|{#1} \right\|_{{#2}}}

\newcommand{\Fnorm}[1]{\lnorm{#1}{\rm F}}
\newcommand{\pth}[1]{\left( #1 \right)}

\newcommand{\iiddistr}{{\stackrel{\text{\iid}}{\sim}}}
\def\simiid{\iiddistr}

\newcommand\indep{\protect\mathpalette{\protect\independenT}{\perp}}
\def\independenT#1#2{\mathrel{\rlap{$#1#2$}\mkern2mu{#1#2}}}
\newcommand{\Bern}{\text{Bern}}
\newcommand{\Poi}{\text{Poi}}

\newcommand{\Iprod}[2]{\langle #1, #2 \rangle}
\newcommand{\indc}[1]{{\mathbf{1}_{\left\{{#1}\right\}}}}

\definecolor{myblue}{rgb}{.8, .8, 1}
\definecolor{mathblue}{rgb}{0.2472, 0.24, 0.6} 
\definecolor{mathred}{rgb}{0.6, 0.24, 0.442893}
\definecolor{mathyellow}{rgb}{0.6, 0.547014, 0.24}

\newcommand{\sfS}{{\mathsf{S}}}

\newcommand{\calY}{{\mathcal{Y}}}

\newcommand{\diverge}{\to \infty}

\newrefformat{eq}{(\ref{#1})}
\newrefformat{thm}{Theorem~\ref{#1}}
\newrefformat{th}{Theorem~\ref{#1}}
\newrefformat{chap}{Chapter~\ref{#1}}
\newrefformat{sec}{Section~\ref{#1}}
\newrefformat{algo}{Algorithm~\ref{#1}}
\newrefformat{fig}{Fig.~\ref{#1}}
\newrefformat{tab}{Table~\ref{#1}}
\newrefformat{rmk}{Remark~\ref{#1}}
\newrefformat{clm}{Claim~\ref{#1}}
\newrefformat{def}{Definition~\ref{#1}}
\newrefformat{cor}{Corollary~\ref{#1}}
\newrefformat{lmm}{Lemma~\ref{#1}}
\newrefformat{prop}{Proposition~\ref{#1}}
\newrefformat{app}{Appendix~\ref{#1}}
\newrefformat{apx}{Appendix~\ref{#1}}
\newrefformat{ex}{Example~\ref{#1}}
\newrefformat{exer}{Exercise~\ref{#1}}
\newrefformat{soln}{Solution~\ref{#1}}

\def\unifto{\mathop{{\mskip 3mu plus 2mu minus 1mu%
	\setbox0=\hbox{$\mathchar"3221$}%
	\raise.6ex\copy0\kern-\wd0%
	\lower0.5ex\hbox{$\mathchar"3221$}}\mskip 3mu plus 2mu minus 1mu}}

\ifx\lesssim\undefined
\def\simleq{{{\mskip 3mu plus 2mu minus 1mu%
	\setbox0=\hbox{$\mathchar"013C$}%
	\raise.2ex\copy0\kern-\wd0%
	\lower0.9ex\hbox{$\mathchar"0218$}}\mskip 3mu plus 2mu minus 1mu}}
\else
\def\simleq{\lesssim}
\fi

\ifx\gtrsim\undefined
\def\simgeq{{{\mskip 3mu plus 2mu minus 1mu%
	\setbox0=\hbox{$\mathchar"013E$}%
	\raise.2ex\copy0\kern-\wd0%
	\lower0.9ex\hbox{$\mathchar"0218$}}\mskip 3mu plus 2mu minus 1mu}}
\else
\def\simgeq{\gtrsim}
\fi

\newtheorem{theorem}{Theorem}
\newtheorem{lemma}[theorem]{Lemma}
\newtheorem{corollary}[theorem]{Corollary}
\newtheorem{coro}[theorem]{Corollary}
\newtheorem{proposition}[theorem]{Proposition}

\theoremstyle{definition}

\newtheorem{remark}{Remark}

\newif\ifmapx
{\catcode`/=0 \catcode`\\=12/gdef/mkillslash\#1{#1}}
\edef\jobnametmp{\expandafter\string\csname dual2_apx\endcsname}
\edef\jobnameapx{\expandafter\mkillslash\jobnametmp}
\edef\jobnameexpand{\jobname}
\ifx\jobnameexpand\jobnameapx
\mapxtrue
\else
\mapxfalse
\fi

\long\def\apxonly#1{\ifmapx{\color{blue}#1}\fi}

\begin{document}
\ifpdf
\DeclareGraphicsExtensions{.pgf,.pdf}
\graphicspath{{figures/}{plots/}}
\fi

\title{Application of information-percolation method to reconstruction problems on graphs}

\author{Yury~Polyanskiy and Yihong Wu\thanks{Y.P.~is with the Department of EECS, MIT, Cambridge, MA, email: \url{yp@mit.edu}. Y.W.~is with
the Department of Statistics and Data Science, Yale University, New Haven, CT, email: \url{yihong.wu@yale.edu}.}}

\date{}

\maketitle

\begin{abstract}
In this paper we propose a method of proving impossibility results based on applying strong data-processing
inequalities to estimate mutual information between sets of variables forming certain Markov random fields. The end
result is that mutual information between two ``far away'' (as measured by the graph distance) variables is bounded by
the probability of the existence of an open path in a bond-percolation problem on the same graph. 
Furthermore, stronger bounds can be obtained by establishing mutual information comparison results with an erasure model on the same graph, with erasure
probabilities given by the contraction coefficients.

As applications, we show that our method gives sharp threshold for partially recovering a rank-one perturbation of a
random Gaussian matrix (spiked Wigner model), yields the best known upper bound on the noise level for group
synchronization (obtained concurrently by Abbe and Boix), and establishes new impossibility result for community detection on the stochastic block model with $k$ communities. 
\end{abstract}

\tableofcontents

\section{Introduction}

As a generalization of ideas of Evans-Schulman~\cite{evans1999signal}, a method for upper-bounding the mutual information between sets of
variables via the probability of the existence of a percolation path was proposed 
by the authors in~\cite[Theorem 5]{PW15-sdpi-tutorial}. This allows one to reuse results on critical threshold for percolation to
show the vanishing of mutual information. The original bound was stated for Bayesian networks (known as directed graphical
models). In this paper we show that similar results can be obtained for certain Markov random fields (\emph{undirected} graphical models) too, especially
those arising in statistical reconstruction problems on graphs such as community detection and group synchronization. 

Our original motivation was to improve the bound on the phase transition threshold in the $\mathbb{Z}_2$-synchronization
on a 2D square grid which appeared in the work of Abbe, Massouli\'e, Montanari, Sly and Srivastava~\cite{abbe2017group}. 
The possibility of such an improvement was anticipated by Abbe and Boix~\cite{Abbe2018-april}, who presented their
work in~\cite{AB18-gsync}, concurrently with the initial circulation of this work. The resulting improvement is stated
below as Corollary~\ref{cor:z2sync} and concides with the result in~\cite{AB18-gsync,abbe2018information}.

The paper is organized as follows. First, we present the idea in its simplest form (binary labels and binary symmetric
channels) in Section~\ref{sec:infoperc}. Second, we extend the method in two
different directions in Section~\ref{sec:gen_perc} to general (non-binary) labels and channels, and in
Section~\ref{sec:nonperc} to non-independent labels. To showcase our general results in Section~\ref{sec:est} 
we consider three applications (of which the first two are chosen following~\cite{AB18-gsync}):
group synchronization, spiked Wigner model and stochastic block model with $k$ blocks.
For the latter our results strengthen (in some regime) the best known impossibility results on correlated (partial)
recovery for $k=3$. 
One of the main technical tools is the strong data processing inequality for mutual information, which is surveyed in the \cite{PW15-sdpi-tutorial};  in \prettyref{app:eta} we provide a quick review emphasizing binary-input channels. We conclude with Appendix~\ref{sec:ab18} comparing our results with the work of Abbe and
Boix~\cite{abbe2018information}.

\section{Information--percolation bound (basic version)}
\label{sec:infoperc}

We start by recalling some basic notions from information theory; cf.~e.g.~\cite{cover}. The mutual information $I(X;Y)$ between random variables $X$ and $Y$ with joint law $P_{XY}$ is $I(X;Y) = D(P_{XY}\|P_X \otimes P_Y)$, where $D(P\|Q)$ is the Kullback-Leibler (KL) divergence between distributions $P$ and $Q$, defined as 
 $D(P\|Q)= \int dP \log \frac{dP}{dQ}$ if $P\ll Q$ and $\infty$ otherwise.
In addition, the $\chi^2$-divergence is defined as $\chi^2(P\|Q)= \int dP (\frac{dP}{dQ}-1)^2$ if $P\ll Q$ and $\infty$ otherwise, and the squared Hellinger distance is $H^2(P,Q)=\int (\sqrt{\frac{dP}{d\mu}}-\sqrt{\frac{dQ}{d\mu}})^2 d\mu$ for any $\mu$ such that $P\ll \mu$ and $Q\ll \mu$.
For discrete $X$, $I(X;Y)=H(X)-H(X|Y)$, the difference between the Shannon entropy of $X$ and the conditional entropy of $X$ given $Y$.

Two properties of mutual information are particularly useful for the present paper:
(a) Chain rule: $I(X;Y,Z) = I(X;Y) + I(X;Z|Y)$, where $I(X;Z|Y)$ is the conditional mutual information.
(b) \emph{Data processing inequality} (DPI): whenever $W\to X\to Y$ forms a Markov chain, we have 
$I(W;Y) \leq I(W;X)$. 
Furthermore, a quantitative version of the DPI is the \emph{strong data processing inequality} (SDPI), 
\begin{equation}\label{eq:sdpi}
I(W;Y) \leq \eta(P_{Y|X}) I(W;X)
\end{equation}
 where $\eta(P_{Y|X}) \in [0,1]$ is called the KL contraction coefficient of the channel. 
For example, if $P_{Y|X}$ is the binary symmetric channel (BSC) with flip probability $\delta$, denoted by $\BSC(\delta)$, that is, $Y= X+ Z \mod 2 \triangleq X\oplus Z$, where $Z \sim \Bern(\delta)$ is independent of $X$, we have $\eta(\BSC(\delta)) = (1-2\delta)^2$. 
For more on SDPI, we refer the reader to the survey \cite{PW15-sdpi-tutorial} and the references therein.

Let $\ER(G,p)$ denote the Erd\"os-R\'enyi random graph on the vertex set $V$, where each edge $e\in E$ is kept independently with probability $p$.
Abbreviate $\ER(K_n,p)$ as $\ER(n,p)$, where $K_n$ is the complete graph on $[n]$.

In this section we consider the following graphical model. Let  $G=(V,E)$ be a simple undirected graph with finite or countably-infinite $V$.
Let $\{X_v: v \in V\}$ be $\iiddistr \Bern(1/2)$ and  
Let $\{Z_e: e \in E\}$ be $\iiddistr \Bern(\delta)$.
For each $e=(u,v)\in E$, let $Y_e = X_u \oplus X_v \oplus Z_{e}$.
For any $S$, let $X_S=\{X_v: v\in S\}$.

\begin{theorem}
\label{thm:perc}
		For any subset $S\subset V$ and any vertex $v\in V$,
	\begin{equation}
I(X_v; X_{S}, Y_E) \leq \perc_G(v, S) \log 2,	
	\label{eq:perc}
	\end{equation}
	where $\perc_G(v,S) = \prob{\text{$v$ is connected to $S$ in $\ER(G,\eta)$}}$, with
	\[
	\eta \triangleq (1-2\delta)^2.
	\]
\end{theorem}

\begin{remark} Notice that right-hand side of~\eqref{eq:perc} can be seen as $I(X_v; X_S, \tilde Y_E)$ where for $e=(u,v)$, $\tilde
Y_e$  is a random variable equal to $X_u\oplus X_v$ with probability $\eta$ and $*$ (erasure) otherwise.
This is not accidental -- it can be shown via~\cite[Prop. 15, 16]{PW15-sdpi-tutorial} that observations over the
erasure channel $\BEC(\eta)$ lead to strictly larger mutual informations: $I(X_{S_1}; Y_E | X_{S_2}) \le
I(X_{S_1}; \tilde Y_E | X_{S_2})$, regardless of the joint distribution $P_{X_V}$. This generalization is pursued in
Section~\ref{sec:nonperc}. \label{rem:perc_erasure}
\end{remark}
 
\begin{proof}
	By the monotone convergence property of mutual information (and probability), it suffices to consider finite graph $G$.

	Let $\bar X_V = \{X_v \oplus1: v\in V\}$. The symmetry of the problem shows that
	$$ (X_V, Y_E) \stackrel{d}{=} (\bar X_V, Y_E)\,.$$
	In particular, we have
	\begin{equation}\label{eq:fpp1}
			I(X_v; Y_E)=0 
	\end{equation}		
	for any $v$.
	
	Fix $V$ and $v\in V$. We induct on the number of edges $|E|$. For the base case of $E = \emptyset$, by the independence of $\{X_v\}$, we have
\[
I(X_v;X_S) =  \indc{v \in S} \log 2 = \perc_G(v,S) \log 2.
\]

	Next suppose \prettyref{eq:perc} holds for all $G' = (V,E')$ with $|E'| < |E|$ and all $S$, i.e.
\begin{equation}
I(X_{z}; X_{S}, Y_{E'}) \leq \perc_{G'}(z, S) \log 2.
\label{eq:induct}
\end{equation}
We now show \prettyref{eq:perc} holds for $E$. Fix $S$. 
Suppose there is no edge in $E$ incident to any vertex in $S$. Then 
both sides of \prettyref{eq:perc} are zero by~\eqref{eq:fpp1}.
Otherwise, there exists an edge $e=(u,z)\in E$ incident to some vertex $z \in S$. 
Set 
$E'=E\setminus e$ and $G'=(V,E')$. 

Next we apply the strong data processing inequality (SDPI) for BSC (see~\cite{PW15-sdpi-tutorial} for a survey on
SDPIs): Note that $Y_e= X_u+X_z+Z_{e}$, where $Z_e\iiddistr \Bern(\delta)$ and independent of $X_V$. Since $z\in S$, conditioned on $(X_S,Y_{E'})$, we have the Markov chain:
$X_v \to X_u \to Y_e$. 
Therefore
\[
I(X_v; Y_e| X_S,Y_{E'}) \leq \eta I(X_v; X_u| X_S,Y_{E'}) .
\]
Adding $I(X_v;X_S,Y_{E'})$ to both sides gives
\[
I(X_{v}; X_S,Y_E) 
\leq 
\eta I(X_{v}; X_S, X_u, Y_{E'}) +
\bar\eta I(X_{v}; X_S, Y_{E'}).
\]
Applying the induction hypothesis \prettyref{eq:induct} to the RHS of the above display, we have:  
\[
I(X_{v}; X_S, Y_E) \leq (\underbrace{\eta \cdot \perc_{G'}(v, S\cup \{u\}) + \bar\eta \perc_{G'}(v, S)}_{=\perc(G)(v,S)} )\cdot \log 2
\]
\end{proof}

\subsection{Simple example of tightness of the bound}\label{sec:simp_perc}

Let $G$ be a complete infinite $d$-ary tree rooted at node $\rho$. Let $S_k$ denote the set of all nodes at depth $k$. Then, by results on 
broadcasting on trees~\cite{EKPS00}, it is easy to see that\footnote{
Indeed, due to the fact that $X_v$'s are iid $\Bern(\frac{1}{2})$, we have $I(X_\rho; Y_E, X_{S_k}) = I(X_\rho; Z_{S_k})$, where for each $u\in S_k$, 
$Z_u = X_u+\sum_{e \in \text{ path from $\rho$ to $u$}} Y_e$. This is precisely the setting of broadcasting on trees, where the label at each node is obtained by passing that of its parent through $\BSC(\delta)$ independently. It was found in \cite{EKPS00} that the cutoff of the total variation $\TV(P_{Z_{S_k}|X_\rho=+}, P_{Z_{S_k}|X_\rho=-})$
 (and equivalently, the mutual information $I(X_\rho; Z_{S_k})$) happens at the threshold $(1-2\delta)^2 d \le 1$.} 
\begin{equation}\label{eq:bcast_trees}
	\lim_{k\to \infty} I(X_\rho; Y_E, X_{S_k}) = \begin{cases} 0, &(1-2\delta)^2 d \le 1\\
						>0, & (1-2\delta)^2 d > 1 
						\end{cases}
\end{equation}						
The bound in Theorem~\ref{thm:perc} is tight in this case in the sense that the right-hand side of~\eqref{eq:perc}
converges to zero if and only if the branching process with offspring distribution $\Binom(d,\eta)$ (with $\eta=(1-2\delta)^2$) is ultimately extinct almost surely, which occurs when $(1-2\delta)^2d \le 1$ by standard results in branching process \cite{AthreyaNey}.

\apxonly{\textbf{TODO:} Need to verify that sums $X_u+\sum_e Y_e$ over all
paths from root to a node $u\in S_k$ are suff. stat. here!}

\section{General version: information percolation}\label{sec:gen_perc}

Consider a bipartite graph $G=(V,W,E)$ with parts $V,W$ and edges $E$, with finite or countably-infinite $V,W,E$. For
any subset $W' \subset W$ we will denote $G[W']$ the induced subgraph on vertices $V\cup W'$.

Let $\{X_v: v \in V\}$ be a collection of independent discrete random variables.
Let $\{Y_w: w\in W\}$ be a collection of random variables conditionally independent given $X_V$ and distributed each as
	\begin{equation}
	Y_w \sim P_{Y_w| X_{N(w)}} \qquad \forall w \in W\,,
	\label{eq:mrf}
	\end{equation}
	where $N(w) \subset V$ denote the neighborhood of $w$ in the bipartite graph $G$.
Let $\eta_w \eqdef \eta_{KL}(P_{Y_w|X_{N(w)}})$ be the SDPI constant corresponding to this
channel~\cite{PW15-sdpi-tutorial}. 

Let $\tilde{G}$ denote the subgraph $G[\tilde W]$ induced by the random subset $\tilde W$, where each vertex $w \in W$ is included in $\tilde W$ independently with
probability $\eta_w$. For a pair of sets $S_1,S_2 \subset V$ we define the average number of vertices in $S_1$ that are
connected to $S_2$:
	$$ \perc_G(S_1,S_2) \eqdef \sum_{v \in S_1} \PP[v \mbox{ is connected to $S_2$ in~} \tilde{G}]\,.$$
We note the following simple identity: if $w$ is such that $N(w) \cap S_2 \neq \emptyset$ then
\begin{equation}\label{eq:gpp1}
	\perc_G(S_1,S_2) = \eta_w \perc_{G[W\setminus w]}(S_1,S_2\cup N(w)) + (1-\eta_w) \perc_{G[W\setminus
	w]}(S_1,S_2)\,.
\end{equation}

To recover the setting of the previous section, where the graph was simple, we can consider the bipartite graph which is the incidence graph
between vertices and edges (in this case the degree of every $w\in W$ is two).

\begin{theorem}
\label{thm:gperc}
	For any subsets $S_1, S_2$ of $V$, we have
	\begin{equation}
I(X_{S_1}; X_{S_2}| Y_W) \leq \perc_G(S_1, S_2) \cdot \sup_{v \in V} H(X_v)\,.
	\label{eq:gperc}
	\end{equation}
\end{theorem}
\begin{remark} Note that $I(X_{S_1}; Y_W)=0$ does not hold even in the setting of the previous section, unless $S_1$ is a
singleton (see~\eqref{eq:fpp1}). Indeed, one may consider the graph $a-b-c$ in the context of \prettyref{thm:perc}. For $S_1=\{a,c\}$, $I(X_{a,c};Y_{ab,bc})\geq I(X_a+X_c; Y_{ab}+Y_{bc}) \geq 1 - h(2\delta(1-\delta))$.
 Thus $I(X_{S_1}; X_{S_2}, Y_W) \neq I(X_{S_1};
X_{S_2}|Y_W)$ and the former does not satisfy the inequality in \prettyref{thm:gperc}.
\end{remark}
\begin{proof} Again, because of the identity
	$$ I(X_{S_1}; X_{S_2} | Y_W) = I(X_{S_1}; X_{S_2}, Y_W) - I(X_{S_1}; Y_W) $$
and the continuity of mutual information and percolation probability we may only consider finite $S_1,S_2,W$.

We will prove \prettyref{eq:gperc} by induction on $|W|$.
Assume that 
$$ H(X_v) \le H_1 \qquad \forall v \in V\,.$$
First, suppose that $W=\emptyset$. We have then:
$$ I(X_{S_1}; X_{S_2}) = \sum_{i\in S_1\cap S_2} H(X_i) \le |S_1 \cap S_2| H_1 = \perc_{G[W]}(S_1,S_2) H_1\,.$$
Next, suppose that we have shown~\eqref{eq:gperc} for all $G[W']$ with $|W'| < |W|$. 
Consider two cases:

Case 1. There does not exist $w \in W$ such that $N(w) \cap S_2 \neq \emptyset$. Then, we have
$$ I(X_{S_1}; X_{S_2} | Y_W) \le I(X_{S_1}, Y_W; X_{S_2}) \le I(X_{S_1}, X_{S_0}; X_{S_2}) \le |S_1\cap S_2| H_1\,,$$
where $S_0 = \cup_{w\in W} N(w)$ and the last equality is due to $S_0 \cap S_2 = \emptyset$. Similarly, we have
$$ \perc_{G}(S_1, S_2) = |S_1 \cap S_2| $$
and~\eqref{eq:gperc} is established.

Case 2. There exists $w \in W$ such that $N(w) \cap S_2 \neq \emptyset$. Let $W' = W\setminus w$. Then we have
\begin{align}
	I(X_{S_1}; X_{S_2}, Y_{W'}, Y_w) &=  I(X_{S_1}; X_{S_2}, Y_{W'}) + I(X_{S_1}; Y_w | X_{S_2}, Y_{W'}) \nonumber \\
				&\le I(X_{S_1}; X_{S_2}, Y_{W'}) + \eta_w I(X_{S_1}; X_{N(w)} | X_{S_2}, Y_{W'}) \nonumber \\
				& = (1-\eta_w) I(X_{S_1}; X_{S_2}, Y_{W'}) + \eta_w I(X_{S_1}; X_{N(w) \cup S_2},
				Y_{W'})\,, \nonumber \\
				& = (1-\eta_w) I(X_{S_1}; X_{S_2}| Y_{W'}) + \eta_w I(X_{S_1}; X_{N(w) \cup S_2}|
				Y_{W'}) + I(X_{S_1}; Y_{W'})\label{eq:gpp2}
\end{align}				
where the inequality is an application of the SDPI, which is justified since given $X_{S_2}, Y_{W'}$ we still have the
Markov chain: $X_{S_1} \to X_{N(w)} \to Y_w$, in view of the definition \prettyref{eq:mrf}.

Subtracting $I(X_{S_1}; Y_{W})$ from both sides of~\eqref{eq:gpp2} we get
\begin{align} I(X_{S_1}; X_{S_2}| Y_{W}) &\le (1-\eta_w) I(X_{S_1}; X_{S_2}| Y_{W'}) + \eta_w I(X_{S_1}; X_{N(w) \cup S_2}|
				Y_{W'}) + I(X_{S_1}; Y_{W'}) - I(X_{S_1}; Y_W) \\
				&\le (1-\eta_w) I(X_{S_1}; X_{S_2}| Y_{W'}) + \eta_w I(X_{S_1}; X_{N(w) \cup S_2}|
				Y_{W'})\,,
\end{align}	
since $I(X_{S_1}; Y_{W'}) \le I(X_{S_1}; Y_W)$ by the monotonicity of the mutual information. From the induction hypothesis
and~\eqref{eq:gpp1} we conclude the proof of~\eqref{eq:gperc}.
\end{proof}

\section{General version: channel comparison}\label{sec:nonperc}

In the setting of Section~\ref{sec:infoperc}, we have imposed the condition~\eqref{eq:fpp1} which implies
$$ I(X_v; X_S, Y_E) = I(X_v; X_S|Y_E) = I(X_v; Y_E|X_S)\,.$$
Consequently, Theorem~\ref{thm:perc} (giving a bound on the first quantity) and Theorem~\ref{thm:gperc} (giving a bound
on the second one) are equivalent when \eqref{eq:fpp1} holds. In fact, Theorem~\ref{thm:gperc} holds in wider
generality. Can we also bound the third quantity? It turns out the answer is yes, and in fact this
generalization allows one to remove the most restrictive condition of Theorem~\ref{thm:gperc} -- the independence of $X_v$'s.
(However, the two theorems bound different quantities.) 
	To focus ideas, we recommend revisiting Remark~\ref{rem:perc_erasure}.

We proceed to describing the setting of the forthcoming more general result. 
Consider a bipartite graph $G=(V,W,E)$ with parts $V,W$ and edges $E$, with finite or countably-infinite $V,W,E$. For
any subset $W' \subset W$, we again denote by $G[W']$ the induced subgraph on vertices $V\cup W'$.

Let $\{X_v: v \in V\}$ be a collection of discrete random variables (not necessarily independent). 
Let $\{Y_w: w\in W\}$ and $\{\tilde Y_w: w\in W\}$ be two collection of random variables each 
conditionally independent given $X_V$ and distributed as
	\begin{align} Y_w &\sim P_{Y_w| X_{N(w)}} \qquad \forall w \in W\,, \\
	  \tilde Y_w &\sim Q_{Y_w| X_{N(w)}} \qquad \forall w \in W
\end{align}	
	where $N(w) \subset V$ denote the neighborhood of $w$ in the bipartite graph.

	We also recall the definition of the less noisy relation: stochastic matrix $Q_{\tilde Y|X}$ is less noisy than $P_{Y|X}$
	if for every distribution $P_{U,X}$ we have
		$$ I(U;Y) \le I(U; \tilde Y) $$
	where mutual informations are computed under the joint distribution
		$$P_{U,X,\tilde Y, Y}(u,x,\tilde y, y) = P_{U,X}(u,x) Q_{\tilde Y|X}(\tilde y|x) P_{Y|X}(y|x)\,. $$
	See~\cite[Theorem 2]{vanDijk97},~\cite[Prop. 14]{PW15-sdpi-tutorial} and~\cite[Theorem 2, Prop.
	8]{MP16-comp-arxiv} for various characterizations of the less noisy relation.

\begin{theorem}\label{thm:compare} 
Assume that for every $w\in W$, the channel $Q_{\tilde Y_w|X_{N(w)}}$ is less noisy than $P_{Y_w |
X_{N(w)}}$. 
Then for any subsets $S_1, S_2 \subset V$, we have
	\begin{equation}\label{eq:compare}
		I(X_{S_1}; Y_E | X_{S_2}) \le I(X_{S_1}; \tilde Y_E | X_{S_2})\,.
\end{equation}	
\end{theorem}

\begin{remark}
\label{rmk:beccompare}
 The connection between Theorems~\ref{thm:compare} and \ref{thm:gperc} arises from~\cite[Proposition 15]{PW15-sdpi-tutorial}: the SDPI constant of the 
channel $P_{Y|X}$ satisfies $\etaKL(P_{Y|X}) \le 1-\delta$ if and only if $P_{Y|X}$ is more noisy than the erasure channel $Q_{\tilde Y|X}$
which outputs $\tilde Y=X$ with probability $1-\delta$ and $\tilde Y=*$ (erasure) otherwise. 
\end{remark}
\begin{remark} One cannot replace the less noisy condition with ``more capable'', a weaker notion (see~\cite{KM75comparison}). Indeed, it is known that erasure channel with probability of erasure $1-h(\delta)$ is 
more capable than $\BSC(\delta)$. But then 
consider the example in Section~\ref{sec:simp_perc}. If the more capable variation
of Theorem~\ref{thm:compare} were true, we would be able to reduce the probability of an open bond from $(1-2\delta)^2$ to
$1-h(\delta)$ and thus contradict~\eqref{eq:bcast_trees}.
\end{remark}

\begin{proof} Conditioning on $X_{S_2}$ we get a Markov chain $X_{S_1} \to X_V \to Y_E $. By~\cite[Prop.
14]{PW15-sdpi-tutorial}, the less noisy relation tensorizes. That is, the channel $X_V \to \tilde Y_E$ is less noisy than
$X_V \to Y_E$. Consequently, we get~\eqref{eq:compare}.
\end{proof}

\section{Applications to reconstruction problems}
	\label{sec:est}
	
	In this section we apply the information-percolation bound to various reconstruction problems on graphs, specifically, 
	$\integers_2$-synchronization, spiked Wigner model, community detection on stochastic block model (SBM) with two, and more than two communities. 
	The first three were considered earlier in \cite{AB18-gsync}, while the fourth application is new and apparently
	not obtainable via methods of~\cite{AB18-gsync,abbe2018information}.

\subsection{Group synchronization over $\mathbb{Z}/2\mathbb{Z}$}

The problem of group synchronization refers to the following: 
Given a graph $G=(V,E)$, let $X_V=\{X_v\}_{v\in V}$ be a collection of independent random variables that are uniformly distributed on some compact group. The goal 
is to recover $X_V$ (up to a global group action which is not identifiable) from pairwise measurements $Y_E=\{Y_{uv}\}_{(u,v)\in E}$, where $Y_{uv}$ is a noisy observation of $X_u^{-1}X_v$.
The paradigm of group synchronization arises in a various applications such as localization, imaging and computer vision (cf.~the references in \cite{abbe2017group}).

The synchronization problem over the $d$-dimensional grid was studied in \cite{abbe2017group} for various groups, focusing on \emph{correlated recovery}, i.e., achieving a reconstruction error that is strictly better than random guessing.
The simplest problem is for the group $\mathbb{Z}/2\mathbb{Z}$, commonly known as $\mathbb{Z}_2$-synchronization, which precisely corresponds to the setting of Section~\ref{sec:infoperc}.
If the observation channel is $\BSC(\delta)$, it is shown in \cite{abbe2017group} that 
correlated recovery is impossible if $1-2\delta \leq \frac{1}{2}$.
Next, we apply the information-percolation method in \prettyref{thm:perc} to improve the threshold
$(1-2\delta)^2\leq \frac{1}{2}$; this result was first announced and proved independently in \cite{abbe2017group}.
To prove the impossibility of the correlated recovery of $X_V$, it suffices to show that for any pair of vertices $u\neq v$, it is impossible to reconstruct the bit $T_{uv} = X_{v} \oplus X_{u}$ better than chance.

\begin{corollary} 
\label{cor:fano}
For any two (possibly non-adjacent) vertices $u,v\in V$, any estimator 
 $\hat T_{uv}=\hat T_{uv}(Y_E)$ satisfies
\begin{equation}\label{eq:fpp2}
		\PP[\hat T_{uv} \neq T_{uv}] 
		\ge {1\over 2} - \sqrt{{1 \over 2\log e} I(X_u;X_v,Y_E)}
		\ge {1\over 2} - \sqrt{{\log 2\over 2\log e} \perc_G(v, u)}
\end{equation}	
Consequently,
\begin{equation}\label{eq:fpp3}
	\frac{1}{|V|^2}\sum_{u,v \in V}	\PP[\hat T_{uv} \neq T_{uv}] 
		\ge {1\over 2} - o(1)
\end{equation}	
provided
\[ 
\sum_{u,v \in V}	 I(X_u;X_v,Y_E) = o(|V|^2)
\quad
\text{or }
\sum_{u,v \in V}	\perc_G(v, u) = o(|V|^2).
\]
\end{corollary}
\begin{remark}
It is clear, from Theorem~\ref{thm:gperc}, that the result above extends to arbitrary channels $P_{Y_e|X_u, X_v}$ for $e=(u,v)$, arbitrary
function $T=T(X_u, X_v)$ and arbitrary (discrete)
$X_v$. The only general requirement we need to impose the validity of~\eqref{eq:fpp1}. The only change is that the first term ${1\over 2}$ in the right-hand side of~\eqref{eq:fpp2} should be replaced
with $1-\max_s \PP[T(X_u, X_v)  = s]$ and $\log 2$ in the denominator inside the square root with $\max_v H(X_v)$. We
put this corollary first, as it originally motivated the writing of this article.
\end{remark}

\begin{proof}
It suffices to show \prettyref{eq:fpp2} as the rest follows from Jensen's inequality.
Next abbreviate $T_{uv}$ as $T$.
Note that
\begin{align*}
I(T;Y_E)
\stepa{\leq} & ~ I(X_u,X_v;Y_E) = I(X_u;Y_E|X_v) + I(X_v;Y_E) \\
\stepb{=} & ~ I(X_u;Y_E|X_v)	\\
\stepc{=} &~ I(X_u;X_v,Y_E)\\
\stepd{\leq} &~ \perc_G(v, u) \log 2,
\end{align*}
where (a) is the data processing inequality for mutual information;
(b) follows from \prettyref{eq:fpp1};
(c) follows from the assumption that $X_u\indep X_v$;
(d) follows from \prettyref{thm:perc}.

On the other hand,  for any estimator $\hat{T}=\hat{T}(Y_E)$, 
let $p=\pprob{\hat{T}=T}$ and $q=\mathbb{Q}[\hat{T}=T]$, where $\mathbb{Q}$ denote the probability measure where $Y_E$ and $T$ are independent. Thus $q \leq P_{\max}(T) \triangleq \max_t \prob{T=t}$. By the data processing inequality and the Pinsker inequality, we have
\[
I(T;Y_E) \geq d(p\|q) \geq 2 \log e (p-q)^2.
\]
Thus,
\[
\pprob{\hat{T}=T} \leq P_{\max}(T) + \sqrt{\frac{\perc_G(v, u) \log 2}{2 \log e}}.
\]
\end{proof}

Using Kesten's result on 2D-square grid percolation \cite{Kesten80}, we get:
\begin{corollary}\label{cor:z2sync} Let $G$ be an infinite $2D$-grid and suppose the goal is to estimate $T_n = X_{0,0}\oplus X_{n,n}$ for
large $n$ given observations of all (infinitely many) edges $Y_e$. If
	$$ (1-2\delta)^2 \le {1\over 2} $$
then for any estimator $\hat T_n = \hat T_n(Y_E)$ we have $\PP[\hat T_n \neq T_n] \to {1\over 2}$.
\end{corollary}

\subsection{Spiked Wigner model}


Consider the following statistical model for PCA:
\begin{equation}
Y = \sqrt{\frac{\lambda}{n}}  XX^\top + W
\label{eq:spike}
\end{equation}
where $X=(X_1,\ldots,X_n) \in \{\pm 1\}^n$ consists of independent Rademacher entries, and $W$ is a Wigner matrix which is symmetric consisting of independent standard normal off-diagonal entries.
This ensemble is known as the spiked Wigner model (rank-one perturbation of the Wigner ensemble).
Observing the matrix $Y$, the goal is to achieve correlated recovery, i.e., to reconstruct $X$ (up to a global sign flip) better than chance, that is, find $\hat X=\hat X(Y) \in \{\pm1\}^n$, such that
\begin{equation}
\liminf_{n \to\infty} \frac{1}{n}\Expect[|\Iprod{X}{\hat X}|] > 0.
\label{eq:corr}
\end{equation}
It is known that for fixed $\lambda$, if $\lambda > 1$, spectral method (taking the signs of the the first eigenvector of $Y$) achieves correlated recovery \cite{BBP05}. Conversely, if $\lambda < 1$, correlated recovery is information-theoretically impossible. 

As the next result shows, applying \prettyref{thm:perc} together with classical results on Erd\"os-R\'enyi graphs immediately yields the optimal threshold previously obtained in \cite[Theorem 4.3]{deshpande2015asymptotic}. Here, $o(1)$ is any vanishing factor so this result is the best possible.
\begin{coro}
\label{cor:spike}	
	Correlated recovery in the sense of \prettyref{eq:corr} is impossible if 
	\begin{equation}
	\lambda \leq 1 + o(1).
	\label{eq:spikethreshold}
	\end{equation}
\end{coro}
\begin{proof}
Note that \prettyref{eq:corr} is equivalent to 
\begin{equation}
\limsup_{n \to\infty} \frac{1}{n^2}\expect{\Fnorm{XX^\top - \hat X\hat X^\top}^2} < 2.
\label{eq:corr2}
\end{equation}

It is clear that the diagonal entries of $Y$ are independent of $X$ and hence the problem reduces to the setting in \prettyref{sec:infoperc} with $G$ being the complete graph on $n$ vertices and $Y_{ij} = \sqrt{\frac{\lambda}{n}}X_iX_j + W_{ij}$ for $i<j$.
Applying \prettyref{thm:perc} together with \prettyref{cor:fano}, we conclude that: for any $i<j$,
\[
\inf_{\hat T_{ij}(\cdot)} \prob{X_iX_j \neq \hat T_{ij}(Y)} \geq \frac{1}{2} - O(\prob{\text{$i$ and $j$ are connected in $\ER(n,\eta)$}}).
\]
where $\eta=\eta(N(-\sqrt{\frac{\lambda}{n}},1),N(\sqrt{\frac{\lambda}{n}},1))  = \frac{\lambda}{n}(1+o(1))$ in view of  \prettyref{eq:eta-awgn}.
Summing over $i \neq j$, we conclude that for any $\hat X= \hat X(Y) \in \{\pm 1\}^n$, 
\begin{align*}
\Expect\Fnorm{XX^\top - \hat X\hat X^\top}^2 
= & ~ 4 \sum_{i \neq j} \prob{X_iX_j \neq \hat X_i\hat X_j} \\
\geq & ~ 	2n^2 - 2 \sum_{i\in[n]} \expect{\text{size of the connected component in $\ER(n,\eta)$ containing $i$} } \\
\geq & ~ 	2n^2 - n \expect{C_{\max}},
\end{align*}
where $C_{\max}$ denotes the size of the largest connected component in the Erd\H{o}s-R\'enyi graph $\ER(n,\eta)$.
Existing results in the random graph theory show that $\Expect[C_{\max}]=o(n)$ whenever $\eta = \frac{1}{n}(1+o(1)$, which implies the impossibility of \prettyref{eq:corr2}.
Specifically, let $\eta = \frac{1}{n^2}(n+s)$, where $s=o(n)$ by assumption.
By monotonicity, it suffices to consider the case of $s = \omega(n^{2/3})$.
By a result of \L{}uczak \cite[Lemma 3]{Luczak90} (see also \cite[Theorem 5.12]{JLR00}), we have
$C_{\max} \leq c_0 s$ with probability at least $1-c_1n^{1/3}s^{-1/2}$ for some universal constants $c_0,c_1$.
Since $C_{\max} \leq n$, this shows $\Expect[C_{\max}] = o(n)$, completing the proof.
\end{proof}

\begin{remark}[Channel universality]
\label{rmk:generalchannel}	
	Consider a more general observation model than \prettyref{eq:spike}: 
	Let $P(\cdot|\theta)$ be a family of conditional distributions parametrized by $\theta \in \reals$, with conditional density $p_{\theta}(\cdot)$ with respect to some reference measure $\mu$.
Given $M = \sqrt{\frac{\lambda}{n}} XX^\top$, we observe the matrix $Y=(Y_{ij})$, where each $Y_{ij}$ is obtained by passing $M_{ij}$ through the same channel independently, with the conditional distribution given by $P_{Y_{ij}|M_{ij}} = P(\cdot|M_{ij})$.
The spiked Wigner model corresponds to the Gaussian channel $P(\cdot|\theta) = N(\theta,1)$.

Under appropriate regularity conditions on the channel, the sharp threshold \prettyref{eq:spikethreshold} is replaced by the following:
\begin{equation}
\lambda \leq \frac{1}{J_0} + o(1)
\label{eq:spike-general}
\end{equation}
where $J_\theta  \triangleq \int (\frac{\partial p_\theta}{\partial \theta})^2 \frac{1}{p_\theta} d\mu$ is the Fisher information.
This follows from the relationship between the contraction coefficient and the Fisher information. To see why this is true intuitively, 
note that $M_{ij} \in \{\pm \epsilon\}$, with $\epsilon \triangleq \sqrt{\frac{\lambda}{n}}$.
Using the characterization \prettyref{eq:etaLC} of the contraction coefficient for binary-input channels, we have 
$\eta = \sup_{\beta \in [0,1]} \LC_{\beta}(p_{\epsilon}\|p_{-\epsilon})$, where  
$\LC_\beta$ is an $f$-divergence\footnote{Recall an $f$-divergence is defined as $D_f(P\|Q) = \Expect_P[f(\frac{dP}{dQ})]$ for convex $f$ with $f(1)=0$ \cite{Csiszar69}.} with $f(x)=f_\beta(x) = \beta\bar\beta \frac{(x-1)^2}{\beta x + \bar\beta}$. 
By the local expansion of $f$-divergence,
we have 
$D_f(P_{\theta-\delta}\|P_\theta) = \frac{f''(1) J_\theta }{2} \delta^2(1+o(1))$ as $\delta\to 0$. 
Note that $f_\beta''(1) = 2 \beta \bar \beta$, maximized at $\beta=\frac{1}{2}$. It follows that 
$\eta = \frac{\lambda J_0 + o(1)}{n}$.
Thus the same percolation bound used in \prettyref{cor:spike} shows that \prettyref{eq:spike-general} implies the
impossibility of correlated reconstruction.
In the positive direction, it was suggested in \cite[Section II-C]{lesieur2015mmse} that spectral method applied to the
score matrix succeeds provided that $\lambda > \frac{1}{J_0}$. In fact, the full mutual information $I(M; Y)$ also undergoes a phase transition at this point, see \cite{krzakala2016mutual} and~\cite{Barbier16_mutual}.
\end{remark}

\subsection{Community detection: two communities}

Consider a complete graph $K_n$ and $X_v \simiid \Bern(1/2)$. Unlike the group-synchronization case, we have the following
observation channel: for each edge $e=(u,v)$ we have
\begin{equation}\label{eq:cd_channel}
	Y_e = \begin{cases} \Bern(p), & X_u=X_v\\
			\Bern(q), & X_u \neq X_v
	\end{cases} 
\end{equation}	
In other words, $Y$ is the adjacency matrix of a random graph (known as the stochastic block model), in which any pair of vertices are connected with probability $p$ if they are from the same community (with the same labels) or with probability $q$ otherwise.

Given the matrix $Y=(Y_{ij})$, the goal is to achieve correlated recovery, that is, estimating the labels up to a global flip better than random guess. In other words, construct $\hat X = \hat X(Y) \in \{0,1\}^n$, such that
\begin{equation}
\limsup_{n\diverge} \frac{1}{n}\Expect[\min\{ d(\hat X, X), n-d(\hat X, X)\}] < \frac{1}{2},
\label{eq:corr-sbm}
\end{equation}
where $d$ denotes the Hamming distance.
Equivalently, the goal is to estimate $\indc{X_i = X_j}$ for any pair $i,j$ on the basis of $Y$ with probability of error
asymptotically (as $n\to \infty$) not tending to $1/2$. The exact region when this is impossible is
known~\cite{mossel2015reconstruction,mossel2013proof}: 
for $p=\frac{a}{n}$ and $q=\frac{b}{n}$ with fixed $a,b$, correlated recovery is possible
if and only if
	$$ {(a-b)^2\over 2(a+b)} > 1. $$

Appying the information-percolation method (namely Theorem~\ref{thm:gperc}) we get the following slightly suboptimal result (see
Fig.~\ref{fig:nms_comp}).

\begin{figure}[t]
\centering
\includegraphics[width=.5\textwidth]{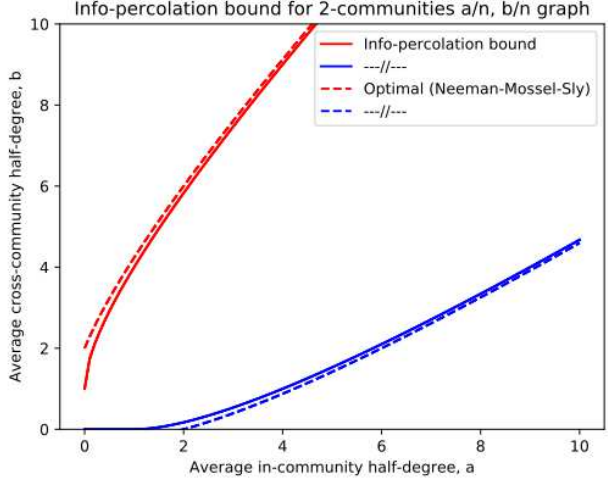}
\caption{Comparing optimal (Mossel-Neeman-Sly~\cite{mossel2015reconstruction}) region with the percolation bound.}
\label{fig:nms_comp}
\end{figure}

\begin{proposition} 
\label{prop:sbm2}
For the binary stochastic block model with edge probabilities $p$ and $q$, for any $i \neq j \in [n]$, the following bound holds non-asymptotically:
\begin{equation}
I(X_i; X_j, Y_E) \leq \prob{\text{$i$ and $j$ are connected in $\ER(n,\eta)$}}
\label{eq:MI-sbm}
\end{equation}
where $\eta = p+q-2pq + 2\sqrt{p(1-p) q(1-q)}$.
Furthermore, if $p=\frac{a}{n}$ and $q=\frac{b}{n}$, then correlated recovery (i.e., \prettyref{eq:corr-sbm}) is impossible if
\begin{equation}
(\sqrt{a} - \sqrt{b})^2 < 1+o(1).
\label{eq:perc-2sbm}
\end{equation}
\end{proposition}
\begin{proof}
The mutual information bound \prettyref{eq:MI-sbm} follows from \prettyref{thm:perc} and the exact expression for the contraction coefficients in \prettyref{eq:eta-binary}, which satisfies
\begin{equation}
\etaKL(\Bern(a/n), \Bern(b/n)) = \frac{(\sqrt{a}-\sqrt{b})^2+o(1)}{n},
\label{eq:etaab}
\end{equation}
where the $o(1)$ terms is uniform in $(a,b)$ in view \prettyref{eq:eta-binary1}.
The remaining proof is the same as \prettyref{cor:spike} using the behavior of the giant component of the \ERG graph.
\end{proof}

\apxonly{REMARK: Let $TEC(p,\delta)$ be a composition of $BSC(\delta)$ followed by $BEC(1-p)$. We can apply
Theorem~\ref{thm:compare} to upper bound $I(X_i; X_j, Y_E) \le I(X_i; X_j, \tilde Y_E)$ by selecting $p=c/n$ and
$\delta$ so that $TEC(c/n,\delta)$ is less-noisy than $\Bern(a/n), \Bern(b/n)$ channel. Now reconstruction problem over
$TEC(c/n,\delta)$ is weakly-solvable iff $(1-2\delta)^2 c > 1$ (converse is broadcasting on trees, achievability -- did
not do.). So we could as if we get better bound in this way. Surprisingly, the answer is no: the best bound on $(a,b)$
is when we compare to vanilla erasure channel $\delta=0$. }

\apxonly{More remarks:
\begin{enumerate}
\item $I_{\chi^2}(X_u; Y_e, X_v) = I_{\chi^2}(X_u,X_v, Y_e) = {1\over n} {(a-b)^2\over 2(a+b)}$
\item $\eta_{KL}(\Bern(1/2), P_{Y_e|X_u,X_v}) > I_{\chi^2}(X_u,X_v, Y_e)$ -- i.e. we cannot tighten the bound by somehow
replacing full contraction coeff with uniform-prior contraction coeff.
\item $\eta_{KL}(\Bern(1/2), P_{Y_e|X_u,X_v}) < \eta_{KL}(P_{Y_e|X_u,X_v})$, so that uniform-prior coeff is strictly
smaller.
\end{enumerate}
}

\subsection{Community detection: $k$ communities}

In the setting of the previous section, suppose now that $X_v \simiid \mathrm{Unif}[k]$, with the same observation channel \prettyref{eq:cd_channel} with $p=\frac{a}{n}$ and $q=\frac{b}{n}$.
This is the stochastic block model with $k$ equal-sized communities, and the notion of correlated recovery is extended as follows:
for any $x,\hat x \in [k]^n$, define the following error metric:
\begin{equation}
d(x,\hat x) \triangleq  \min_{\pi \in S_k} \frac{1}{n}\sum_{i\in[n]} \indc{x_i \neq \pi(\hat x_i)}
\label{eq:lossd}
\end{equation}
that is, the number of classification errors up to a global permutation of labels.
We say that correlated recovery is possible if there exists a (sequence of) estimator $\hat X \in [k]^n$ that outperforms random guessing, i.e., 
\begin{equation}
\limsup_{n\diverge} \eexpect{ d(X,\hat X)} < \frac{k-1}{k}.
\label{eq:corr-sbmk}
\end{equation}
For $k \geq 3$, the sharp threshold is not known.
In terms of the impossibility result, the best known sufficient condition is \cite[Theorem 1]{banks2016information}
\begin{equation}
\frac{(a-b)^2}{a+(k-1)b} < \frac{2k \log(k-1)}{k-1}.
\label{eq:banks}
\end{equation}

Now, it turns out that applying \prettyref{thm:perc} would only yield a $k$-independent bound~\eqref{eq:perc-2sbm}. To
get an improved estimate, instead, we use the comparison theorem with the erasure model in
\prettyref{thm:compare} and then show the impossibility of reconstruction on the corresponding erasure model. The
threshold is given by \prettyref{eq:perc-ksbm} in the next proposition and the numerical comparison with the bound of
\prettyref{eq:banks} is shown in \prettyref{fig:banks_comp}. For $k=3$, \prettyref{eq:perc-ksbm} improves over
\prettyref{eq:banks} in some regime but not for $k=4$. For large $k$, \prettyref{eq:perc-ksbm} is suboptimal by a logarithmic factor.

\begin{figure}[t]
\centering
\includegraphics[width=.45\textwidth]{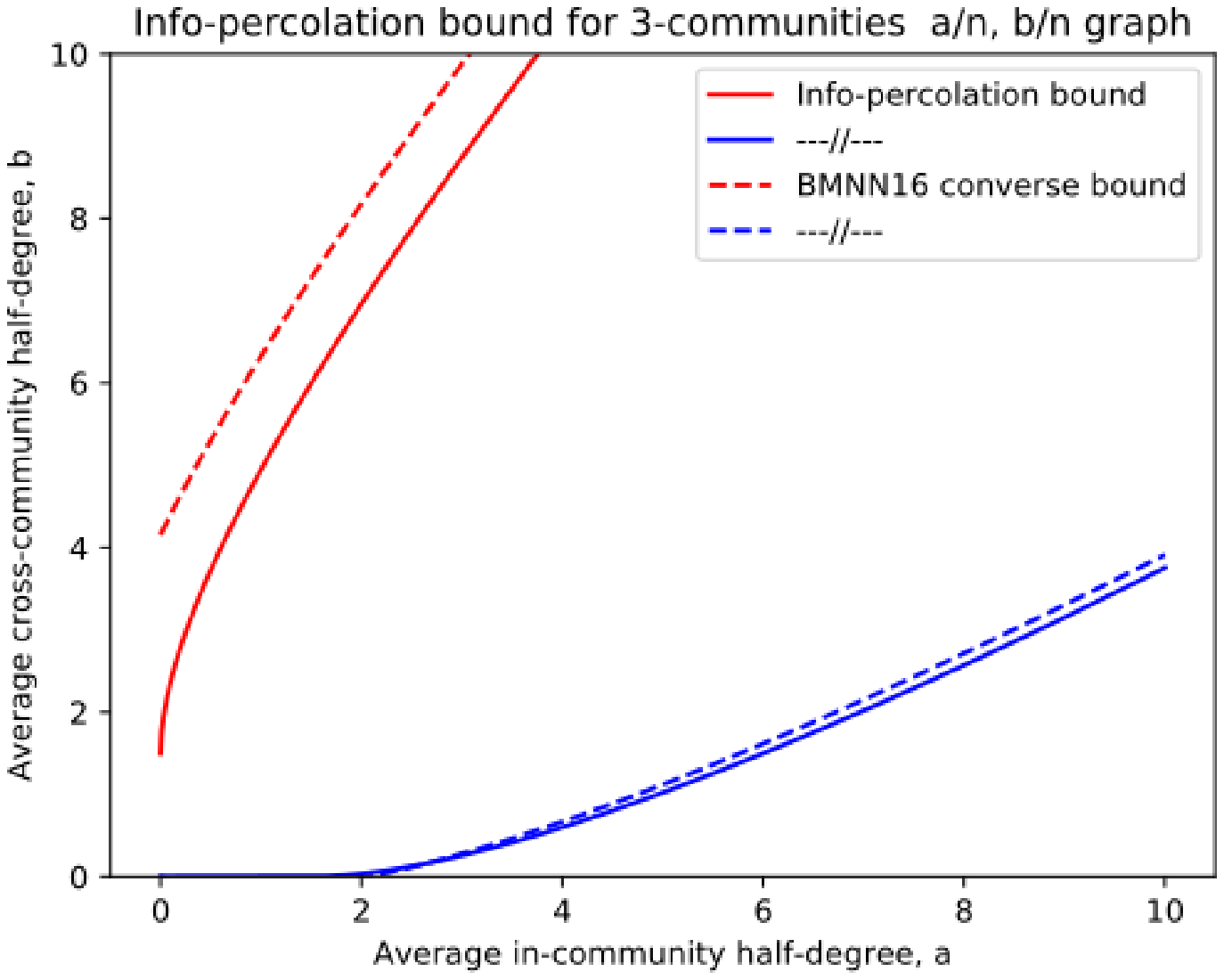}\hskip .1\textwidth
\includegraphics[width=.45\textwidth]{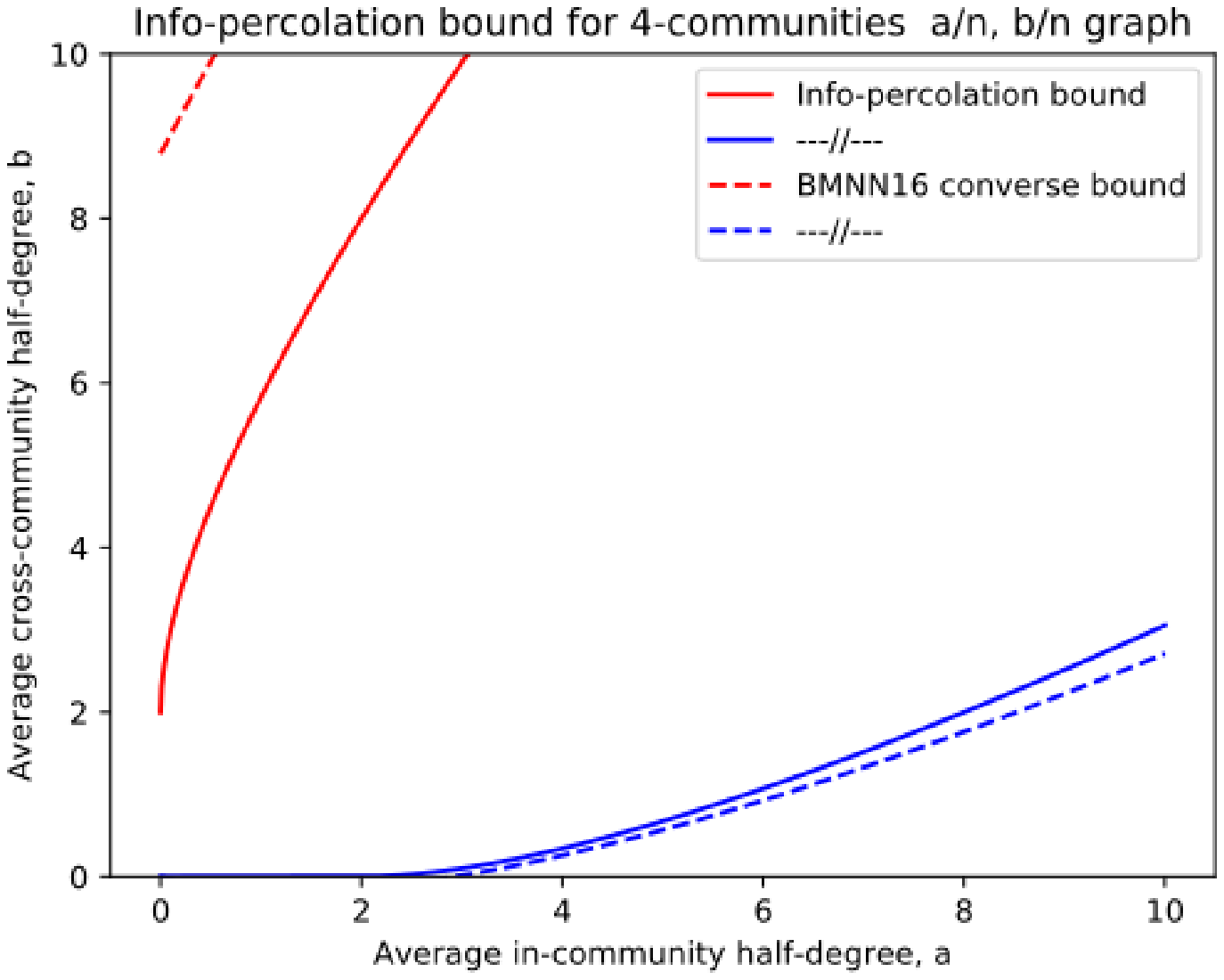}
\caption{Comparing the inner (impossibility) bound of~\cite{banks2016information} with 
Prop.~\ref{prop:ksbm} for $k=3$ and $k=4$ communities. For $k=3$, Prop.~\ref{prop:ksbm} improves the state of the art.}
\label{fig:banks_comp}
\end{figure}

\begin{proposition}\label{prop:ksbm}
Correlated recovery in the sense of \prettyref{eq:corr-sbmk} is impossible if
	\begin{equation}
	(\sqrt{a} - \sqrt{b})^2 \leq \frac{k}{2}.
	\label{eq:perc-ksbm}
	\end{equation}	
\end{proposition}
\begin{proof}
We start by setting up the mutual comparison with the corresponding model per \prettyref{thm:compare}.
Let $\eta = \frac{(\sqrt{a}-\sqrt{b})^2+o(1)}{n}$ be given in \prettyref{eq:etaab}.
Define the corresponding erasure model on the same graph: for each $(u,v)\in \binom{n}{[2]}$, let $\tilde Y_{uv} = \indc{X_u = X_v}$ with probability $\eta$ and 
$\tilde Y_{uv} = ?$ with probability $1-\eta$ independently. 
Equivalently, the reconstruction problem under the erasure model can be phrased as follows. 
Let $G=([n],E)$ denote an Erd\"os-R\'enyi graph $\ER(n,\eta)$ independent of $X$.
Then for each $(u,v) \in E$, we observe a deterministic function $\tilde Y_{uv} = \indc{X_u = X_v}$.
 By Theorem~\ref{thm:compare} 
and \prettyref{rmk:beccompare}, we have the following comparison result: for any $S \subset [n]$, 
\begin{equation}
I(X_S; Y) \leq I(X_S; \tilde Y).
\label{eq:compare-k}
\end{equation}
By symmetry, $I(X_S; \tilde Y)$ only depends on $|S|$. Next we assume $S=[m]$ and show that for any fixed $m$,
\[
I(X_S; \tilde Y) = o(1), \quad n\to\infty
\]
under the condition that $(\sqrt{a} - \sqrt{b})^2 \leq {k\over 2}$.

By the chain rule, we have
\begin{align}
I(X_S; \tilde Y)
= & ~ I(X_1;\tilde Y) + I(X_2;\tilde Y|X_1) + \ldots I(X_m;\tilde Y|X_1,\ldots,X_{m-1}) \nonumber \\
= & ~ \sum_{u=2}^{m} I(X_u;X_1,\ldots,X_{u-1},\tilde Y), \label{eq:chainrule}
\end{align}
where we used the fact that $X_i$'s are independent and $I(X_1;\tilde Y)=0$.

Next using the local tree structure of $G$, we show that for each $u$, $I(X_u;X_1,\ldots,X_{u-1},\tilde Y)=o(1)$.
Condition on the realization of $G$. Fix $t$ to be specified later. 
 Let $G_u^t$ denote the $t$-hop neighborhood of $u$. Let $R$ be the boundary of $G_u^t$, i.e., the 
set of vertices that are at distance $t$ to $u$.
For any $v$ whose distance to $u$ exceeds $t$, 
$R$ forms a cut separating $u$ and $v$ in the sense that any path from $u$ to $v$ passes through $S$.
Then for any set of vertices $U$ outside the $t$-hop neighborhood of $r$, 
 we have
\begin{equation}
I(X_u;X_{U},\tilde Y_E) \leq I(X_u;X_R,\tilde Y_E) = I(X_u;X_R,\tilde Y_{\leq t}),
\label{eq:MIcut}
\end{equation}
where $\tilde Y_{\leq t} \triangleq \tilde Y_{E(G_u^t)}$.
Indeed, the first inequality follows from the fact that
$X_u \to X_R \to X_{S'}$ forms a Markov chain conditioned on $\tilde Y_E$, and the 
second inequality follows from the independence of $X_u$ and $Y_{E(G) \backslash E(G_u^t)}$ conditioned on the $(X_R,Y_{\leq t})$.

By \cite[Proposition 12]{PW14a}, since $X_u$ only takes $k$ values, we can bound the mutual information by the total variation as follows:
\begin{equation}
I(X_u;X_R,\tilde Y_{\leq t})
\leq \log (k-1) T(X_u;X_R,\tilde Y_{\leq t}) + h(T(X_u;X_R,\tilde Y_{\leq t}))
\label{eq:Tinfo}
\end{equation}
where 
$h(x) \triangleq x \log \frac{1}{x}+(1-x)\log \frac{1}{1-x}$, and  
\begin{equation}
T(X_u;X_R,\tilde Y_{\leq t}) \triangleq 
\Expect[\TV(P_{X_R,\tilde Y_{\leq t}|X_u}, P_{X_R,\tilde Y_{\leq t}})] \leq \max_{x,x'\in[k]} \TV(P_{X_R,\tilde Y_{\leq t}|X_u=x}, P_{X_R,\tilde Y_{\leq t}|X_u=x'})
\label{eq:T-maxTV}
\end{equation}
where the last inequality follows from the convexity of the total variation.

Now choose $t=t_n$ such that $t=\omega(1)$ and $t=o(\log n)$.
We show that
\begin{equation}
\tau \triangleq \max_{x,x'\in[k]} \TV(P_{X_R,\tilde Y_{\leq t}|X_u=x}, P_{X_R,\tilde Y_{\leq t}|X_u=x'}) = o(1).
\label{eq:main-commg-tv}
\end{equation}
To this end, let $T_u^t$ denote a depth-$t$ Galton-Watson (GW) tree rooted at $u$ with offspring distribution $\Poi(d)$, with $d\triangleq n \eta$ is at most a constant by assumption. By the locally tree-like property of the Erd\"os-R\'enyi graph (see, e.g., \cite[ Proposition 4.2]{mossel2015reconstruction} with $p=q$), there exists a coupling between $T_u^t$ and $G_u^t$ such that $\prob{G_u^t=T_u^t} = 1-o(1)$.
In the sequel we condition on the event of $G_u^t=T_u^t$
In particular, by standard results in branching process \cite{AthreyaNey}, 
the expected number of $i$th progeny is $d^i$ and hence 
the expect size of the $t$-neighborhood of $u$ is $\frac{d^{t+1}-1}{d-1}$. By the Markov inequality,
the size of the $t$-neighborhood of $u$ is at most $M \triangleq (Cd)^t = n^{o(1)}$ with probability $1 -o(1)$.
In other words, the majority of $v$ are outside the $t$-neighborhood of $u$.
Next we conditioned on the event $G_u^t=T_u^t$ and abbreviate $T_u^t$ as $T$. For each $x \neq x'$, we construct a coupling 
$\{X_v^+,X_v^-: v \in V(T)\}$ and $\{Y_e: e \in E(T)\}$ 
so that $(X_{V(T)}^+,Y_{E(T)})$ and $(X_{V(T)}^-,Y_{E(T)})$ are distributed as the law 
of $(X_{V(T)},Y_{E(T)})$ conditioned on the root $X_u=x$ and $X_u=x'$, respectively.
The coupling is defined inductively as follows:
First set $X_u^+=x$ and $X_u^-=x'$. Next we generate each layer of observations recursively as follows:
Given all the $X_v$'s and $Y_e$'s up to depth $k$, draw $Y_e=\Bern(1/k)$ independently for all edges between the $k$th and the $(k+1)$th layer. For each edge $e=(i,j)$ so that $i$ is on $k$th layer and $j$ is on $(k+1)$th layer, 
if $X^+_i=X^-_i$, we couple all observations on the subtree rooted at $i$ together, that is, set $X^+_j=X^-_j=X^+_i$ if $Y_e=1$  and $X^+_j=X^-_j=R$ if $Y_e=0$ where $R$ is drawn uniformly at random from $[k]\setminus \{X^+_i\}$;
if $X^+_i\neq X^-_i$, we proceed as follows:
\begin{itemize}
	\item if $Y_e = 1$, set $X_{j}^+ = X_{i}^+$ and $X_j^- = X_i^-$.
	\item if $Y_e = 0$, with probability $\frac{k-2}{k-1}$, set $X^+_j = X^-_j = R$ with
	$R$ drawn uniformly at random from $[k] \setminus \{X^+_i, X^-_i\}$, and with probability $\frac{1}{k-1}$ set $X^+_j = X^-_i$, and $X^-_j = X^+_i$.
\end{itemize}
Note that 
for each $i$ and each of its child $j$, we have 
\[
\prob{X_j^+ \neq X_j^- |X_i^+ \neq X_i^-} = \prob{Y_e = 1} + \prob{Y_e = 0} \frac{1}{k-1} = \frac{2}{k}.
\]
Thus, the number of uncoupled pairs $(X_i^+,X_i^-)$ evolves as a GW tree with offspring distribution $\Poi(\frac{2d}{k})$, which dies out if $\frac{2d}{k} \leq 1$ (see, e.g., \cite[Theorem 1]{AthreyaNey}), in which case we have $\TV(P_{X_{V(T)},Y_{E(T)}|X_u=x}, P_{X_{V(T)},Y_{E(T)}|X_u=x'}) \leq \prob{X_R^+ \neq X_R^-} \to 0$, as $t\to\infty$.
This completes the proof of  \prettyref{eq:main-commg-tv}.

Combining \prettyref{eq:Tinfo}--\prettyref{eq:main-commg-tv}, we have
\[
I(X_u; X_1,\ldots,X_{u-1},\tilde Y) \leq \log(k-1) \tau + h(\tau) + (1- \prob{E \cap E'})\log k 
\]
where $E = \{G_u^t=T_u^t, |V(T_u^t)| \leq M \}$, $M=(Cd)^t=n^{o(1)}$, 
and $E'$ denotes the event that $1,\ldots, u-1$ are all outside the $t$-hop neighborhood of $u$.
We have already shown that $\tau=o(1)$ and $\prob{E}=1-o(1)$. Furthermore, by symmetry
$\prob{E'} =  \frac{M-1}{n-1} \cdots \frac{M-u}{n-u} \geq (\frac{M-m}{n-m})^m = 1-o(1)$.
To summarize, we have shown that $I(X_u; X_1,\ldots,X_{u-1},\tilde Y) = o(1)$ and, in view of \prettyref{eq:chainrule},
\begin{equation}
I(X_S; \tilde Y) = o(1)
\label{eq:IXSY}
\end{equation}
for $S=[m]$ and hence any $S \in \binom{[n]}{m}$.

Finally, using \prettyref{eq:IXSY} for appropriately chosen $m$, we show the impossibility of the correlated recovery \prettyref{eq:corr-sbmk}.
First of all, note that for any fixed $x,\hat x\in[k]^n$ and any $m\in [n]$ we have
\begin{equation}
d(x,\hat x)
\geq \Expect_{S}[d(x_{\sfS},\hat x_{\sfS})] \label{eq:davg}
\end{equation}
where ${\sfS}  \sim \Unif(\binom{[n]}{m})$ and recall that for any $S$, we have 
$d(x_S,\hat x_S) = \frac{1}{|S|}  \min_{\pi \in S_k} \sum_{i\in S} \indc{x_i \neq \pi(\hat x_i)}$ per \prettyref{eq:lossd}. The inequality \prettyref{eq:davg} simply follows from
\begin{align}
d(x,\hat x)
= & ~ \min_{\pi \in S_k} \probs{x_I \neq \hat x_{\pi(I)}}{I \sim \Unif([n])}	\nonumber \\
= & ~ \min_{\pi \in S_k} \Expect_{\sfS \sim \Unif(\binom{[n]}{m})} \probs{x_I \neq \hat x_{\pi(I)}}{I \sim \Unif({\sfS})}	\nonumber \\
= & ~ \Expect_{{\sfS}} \min_{\pi \in S_k} \probs{x_I \neq \hat x_{\pi(I)}}{I \sim \Unif({\sfS})}	\nonumber \\
\geq & ~ \Expect_{{\sfS}} [d(x_{\sfS},\hat x_{\sfS})] \nonumber.
\end{align}
Fix a constant $m$ independent of $n$. 
For any estimator $\hat X=\hat X(Y) \in [k]^n$, applying \prettyref{eq:davg} yields
\begin{equation}
\Expect[d(X_{\sfS},\hat X_{\sfS})] \leq \Expect[d(X,\hat X)], \label{eq:davg2}
\end{equation}
where ${\sfS}$ is a random uniform $m$-set independent of $X,\hat X$.

By the data processing inequality, we have for any fixed $S$,
\[
I(X_S; \hat X_S) \leq I(X_S; Y) \overset{\prettyref{eq:compare-k}}{\leq} I(X_S; \tilde Y) \overset{\prettyref{eq:IXSY}}{=} o(1).
\]
By Pinsker's inequality, we have 
$\TV(P_{X_S, \hat X_S}, P_{X_S} \otimes P_{\hat X_S}) \leq \sqrt{2 I(X_S; \hat X_S)} = o(1)$.
Note that the loss function $d$ defined in \prettyref{eq:lossd} is bounded by one.
Thus 
\begin{equation}
\Expect[d(X_S,\hat X_S)] \geq \Expect[d(X_S,Z_S)] - \TV(P_{X_S, \hat X_S}, P_{X_S} \otimes P_{\hat X_S}) 
= \Expect[d(X_S,Z_S)] + o(1), 
\label{eq:ddTV}
\end{equation}
where $Z_S$ has the same distribution as $\hat X_S$ and is independent of $X_S$.
By \prettyref{lmm:randomguess} at the end of this subsection, we have
\begin{equation}
\Expect[d(X_S,Z_S)] \geq \pth{\frac{k-1}{k} - m^{-1/3}}(1-k! e^{-2m^{1/3}}).
\label{eq:randomguess2}
\end{equation}
Combining \prettyref{eq:davg2}, \prettyref{eq:ddTV} and \prettyref{eq:randomguess2}, sending $n\to\infty$ followed by $m \to \infty$, we arrive at
\[
\liminf_{n\diverge} \eexpect{ d(X,\hat X)} \geq \frac{k-1}{k}.
\]
This completes the proof of the proposition.
\end{proof}

\begin{lemma}
\label{lmm:randomguess}	
	Let $X$ be uniformly distributed on $[k]^m$ and $Z$ is independent of $X$ with an arbitrary distribution on $[k]^m$. 
For the loss function in \prettyref{eq:lossd}, we have\footnote{Note that for any fixed $k,m$ and any string $x,z\in [k]^m$, we can always outperform random matching, i.e., $d(x,z) < \frac{k-1}{k}$. The point of \prettyref{eq:randomguess} is that this improvement is negligible for large $m$.}
	\begin{equation}
	d(X,Z) \geq \frac{k-1}{k} - m^{-1/3}
	\label{eq:randomguess}
	\end{equation}
	with probability at least $1-(k! e^{-2m^{1/3}})$.
\end{lemma}
\begin{proof}
	For each fixed $\pi$, the Hamming distance $d_H(X,\pi(Z))\sim \Binom(m,\frac{k-1}{k})$. From Hoeffding's inequality we have
	$$ \PP[d_H(X,\pi(Z) < {k-1\over k} - \delta] \le e^{-2m \delta^2}\,,$$
	and from the union bound
	$$ \PP[\min_\pi d_H(X,\pi(Z) < {k-1\over k} - \delta] \le k! e^{-2m \delta^2}\,.$$
	Setting $\delta = m^{-1/3}$ completes the proof.
\end{proof}

\begin{remark} In the above proof we considered the problem of reconstructing the root $X_u$ variable of a Galton-Watson
tree with the average degree $d$, where the vertex variables are iid and unifrom on $[k]$, and the edge variables are given  $Y_{i,j}
= \indc{X_i=X_j}$ for each edge $i,j$. The reconstruction of $X_u$ is based on the values of all $Y_{i,j}$ and all vertex variables at an arbitrary deep layer of the tree. We have shown
the reconstruction is impossible (unable to outperform random guessing) if $d\le{k\over2}$. At the same time, clearly 
reconstruction is possible if $d\ge k$ (in which case there is an arbitrarily long path of edges with $Y_{i,j}=1$ starting
from the root). So what is the exact threshold? A work in progress~\cite{GuPo2019-draft} shows a much improved bound,
namely that reconstruction is impossible if
	$$ d < f(k) \eqdef \left( {\log k - \log (k-1)\over \log k} {k-1\over k} + {1\over k}\right)^{-1} = k - (1+o(1)){k\over \log k}\,.$$
Using this bound in place of $d<k/2$ it follows that correlated recovery in a $k$-SBM is not possible if 
\begin{equation}\label{eq:gupo_1}
		(\sqrt{a}-\sqrt{b})^2 < f(k)\,.
\end{equation}	
	This improves~\eqref{eq:banks} for all $k\geq3$ in some range of $a,b$. The work~\cite{GuPo2019-draft} presents
	further improvements to~\eqref{eq:gupo_1} based on applying SDPIs directly to an equivalent Potts model on a
	tree.
\end{remark}

\appendix
\section{Contraction coefficients of some binary-input channels}
	\label{app:eta}

Consider an arbitrary channel $P_{Y|X}$. Denote the contraction coefficient, defined as the best constant
in~\eqref{eq:sdpi}, by $\etaKL(P_{Y|X})$. It has an equivalent characterization:
 \begin{equation}\label{eq:etakl_d}
 	\etaKL(P_{Y|X}) = \sup_{\pi_X \neq \pi'_X} {D(Q_Y \| Q_{Y}')\over D(\pi_X \| \pi_X')}\,,
\end{equation} 
where $Q_Y$ and $Q'_Y$ are the distributions induced by $\pi_X$ and $\pi_X'$, respectively.

Consider a binary input channel $P_{Y|X}$, where $P_{Y|X=0}=P$ and $P_{Y|X=1}=Q$. 
Then we can write $\etaKL(P_{Y|X}) = \etaKL(P,Q)$, for convenience.
The following representation is given in \cite[Proof of Theorem 21]{PW15-sdpi-tutorial} in terms of the Le Cam divergence:
	\begin{equation}
	\etaKL(P,Q) = \sup_{\beta \in [0,1]} \underbrace{\beta\bar\beta  \int \frac{(P-Q)^2}{\beta P + \bar\beta Q}}_{\triangleq \LC_{\beta}(P\|Q)},
	\label{eq:etaLC}
	\end{equation}	
	where we denote $\bar\beta=1-\beta$. 
	For example, for a binary-input binary-output channel, direct calculation gives
	\begin{align}
	\etaKL(\Bern(p),\Bern(q)) 
	= &~ p+q-2pq -2\sqrt{p\bar p q \bar q} 	\label{eq:eta-binary} \\
	 \leq &~ (\sqrt{p}-\sqrt{q})^2 + 2\sqrt{pq}(p+q) 	\label{eq:eta-binary1x} 
	\end{align}
	In particular, for the $\BSC(\delta)$ we have $q=1-p=\delta$ and $\etaKL(\BSC(\delta))=(1-2\delta)^2$.

	It is further shown in \cite[Theorem 21]{PW15-sdpi-tutorial} that squared Hellinger distance determines the contraction coefficient of binary-input channel up to a factor of two:	
	\begin{equation}
\frac{H^2(P,Q)}{2} \leq \eta(\{P,Q\}) \leq H^2(P,Q).
	\label{eq:etaPQH}
\end{equation}
Thus, we have
\begin{align}
\etaKL(\Bern(a/n),\Bern(b/n)) \leq & ~ \frac{(\sqrt{a}-\sqrt{b})^2+o(1)}{n}, \quad n\diverge	\label{eq:eta-binary1}\\
\etaKL(N(-\delta,1), N(\delta,1)) \leq & ~ \delta^2(1+o(1)), \quad \delta \to 0. \label{eq:eta-awgn}
\end{align}

For binary-input channels, the SDPI constant can be related to the following $\chi^2$-mutual information:
\begin{equation}\label{eq:ichi2_def}
		I_{\chi^2}(X;Y) \triangleq \chi^2(P_{XY} \| P_X \otimes P_Y)
\end{equation}	
and notice that if $X\sim \Bern(1/2)$ then
	$$ I_{\chi^2}(X;Y) = \LC_{1/2}(P\|Q)=\int \frac{(P-Q)^2}{2(P+Q)}\,. $$
Hence from~\eqref{eq:etaLC} we have 
\begin{equation}\label{eq:etakl_generic}
		\etaKL(P_{Y|X}) \ge I_{\chi^2}(X;Y)\,.
\end{equation}	
Furthermore, under a symmetry assumption, \prettyref{eq:etakl_generic} holds with the equality as we show next.

A binary-input channel $P_{Y|X}$ is called symmetric (often called a BMS channel in the information theory literature \cite{richardson2008modern}) if there exists a measurable involution $T: \calY\to
\calY$ such that $P_{Y|X=0}(T^{-1} A) = P_{Y|X=1}(A)$ for all measurable subsets $A \subset \calY$. For such a channel,
we have that
\begin{equation}\label{eq:etakl_BMS}
		\eta_{KL}(P_{Y|X}) = I_{\chi^2}(X;Y)\,,\qquad X \sim \Bern(1/2)\,,
\end{equation}
Indeed, for the special case of $\BSC(\delta)$, both sides are equal to $(1-2\delta)^2$ by an explicit calculation. 
In general, a well-known decomposition result (cf.~\cite[Lemma 4.28]{richardson2008modern}) shows that any BMS $P_{Y|X}$ can be represented as a mixture of $\BSC$'s.
Namely, we can equivalently think of the action of the channel $P_{Y|X}$ as first generating a random variable $\Delta \in
[0,1]$ according to a fixed distribution $P_\Delta$, passing $X$ through $\BSC(\Delta)$ to obtain $\tilde Y$, and then outputting both $\Delta$ and $\tilde Y$. With this
model, we have $Y=(\Delta, \tilde Y)$ and with $X\sim \Bern(1/2)$ 
	$$ I_{\chi^2}(X; Y) = I_{\chi^2}(X; \Delta, \tilde Y) 
	= \EE[(1-2\Delta)^2]\,.$$
Next, fix two distributions $\pi=\Bern(a)$ and $\pi'=\Bern(a')$. Let $Q=Q_{\Delta, \tilde Y}$ and $Q'=Q_{\Delta, \tilde Y}'$
be the corresponding distributions produced at the output of $P_{Y|X}$. Note that conditioned on $\Delta=\delta$ we have
by the SDPI~\eqref{eq:etakl_d} for the $\BSC(\delta)$:
$$ D(Q_{\tilde Y|\Delta = \delta}\| Q'_{\tilde Y|\Delta=\delta}) \le (1-2\delta)^2 D(\pi \| \pi')\,.$$
Since the marginal distribution of $\Delta$ is the same under $Q$ and $Q'$, taking expectation over $\Delta$ yields
$$ D(Q\|Q') = \EE_{\Delta \sim P_{\Delta}}[D(Q_{\tilde Y|\Delta}\| Q'_{\tilde Y|\Delta})] \le \EE[(1-2\Delta)^2] D(\pi
\| \pi')\,.$$
Therefore, from~\eqref{eq:etakl_d} we get that
$$ \etaKL(P_{Y|X}) \le \EE[(1-2\Delta^2)] = I_{\chi^2}(X; Y)\,.$$
Together with~\eqref{eq:etakl_generic} this completes the proof of~\eqref{eq:etakl_BMS}.

\section{Comparison with~\cite{abbe2018information}}\label{sec:ab18}

The first bound on $\mathbb{Z}_2$-synchronization threshold over a 2D-square grid was obtained in~\cite{abbe2017group}
by leveraging a standard coupling technique, in which the action of the $\BSC(\delta)$ is modeled as passing a bit
uncorrupted with probability $1-2\delta$ or rerandomizing it otherwise. A natural argument
then shows that on an arbitrary lattice the $\mathbb{Z}_2$-synchronization is impossible
whenever $(1-2\delta)$ is smaller than the bond-percolation threshold of the lattice. 

The present work sprang from the remark of E.~Abbe~\cite{Abbe2018-april}, suggesting that an improved estimate on this threshold is 
possible. In the previous work~\cite{PW15-sdpi-tutorial} of the authors, a general technique is developed for showing vanishing of
mutual information in a network of $\BSC(\delta)$-channels whenever $(1-2\delta)^2$ is below the 
vertex-percolation threshold. While the Bayesian network
setup of~\cite{PW15-sdpi-tutorial} is not directly applicable to the setting of group synchronization, the 
method (of induction on the number of edges) does apply. This lead us to Theorem~\ref{thm:perc}, which was disseminated slightly
prior to the talk~\cite{AB18-gsync} presenting a similar result (subsequently published as~\cite{abbe2018information}).
Both our Theorem~\ref{thm:perc} and~\cite{abbe2018information} yield the same threshold for
$\mathbb{Z}_2$-synchronization on a 2D-square grid, cf. Corollary~\ref{cor:fano}.

The main result of~\cite{abbe2018information} is the following. Consider the setting of Theorem~\ref{thm:gperc} and
assume in addition 
\begin{enumerate}
	\item that each label $X_v$ is binary and unbiased: $X_v \sim \Bern(1/2)$;
	\item that each $w\in W$ has degree 2;
	\item that each channel $P_{Y_w|X_{N(w)}}$ has the following special form
		$$ Y_w \sim Q_w(\cdot | X_u \oplus X_v)\,,$$
		where $N(w) = \{u,v\}$ and $Q_w(\cdot|\cdot)$ is a binary-input symmetric channel (BMS).
\end{enumerate}
	Then
	$$ I_{\chi^2}(X_u; X_S, Y_W) \le \perc_G(v, S)\,,$$
	where $I_{\chi^2}$ was defined in~\eqref{eq:ichi2_def} and $\perc_G(v,S)$ is a probability of existence of an open path from $u$ to $S$ if each vertex $w\in W$
	is retained with probability $I_{\chi^2}(X_{N(w)}; Y_w)$.

	In view of \eqref{eq:etakl_BMS} and the bound $I(X_u; X_S, Y_W) \le \log (1+I_{\chi^2}(X_u; X_S, Y_W)) \le \log e \cdot
	I_{\chi^2}(X_u; X_S, Y_W)$,  we see that indeed the result
	of~\cite{abbe2018information} is a special case of Theorem~\ref{thm:gperc}.

Notably the proof in~\cite{abbe2018information} also proceeds by induction on the number of edges (i.e.~on the size of
$W$), similar to our
proofs of Theorems~\ref{thm:perc}-\ref{thm:gperc} and~\cite[Theorem 5]{PW15-sdpi-tutorial}. Indeed, suppose that the
result has been shown for $W'$ and $W=W' \cup\{w_0\}$. Suppose also $N(w_0) = \{u_0,v_0\}$, and in addition that
$Q_{w_0} = \BSC(\delta_0)$ (this assumption is easy to remove by a separate argument). Then \cite{abbe2018information}
exploits the extra structure imposed by the assumptions above and directly computes 
\begin{equation}\label{eq:ab18_st}
		I_{\chi^2}(X_u; X_v, Y_{W}) = I_0 + (I_1 - I_0) (1-2\delta_0)^2 h((1-2\delta_0))^2\,,
\end{equation}	
	where $h:[0,1]\to [0,1]$ is a non-decreasing function that depends only on $W'$ and $\{Q_w, w\in W'\}$, $I_0 = I_{\chi^2}(X_u; X_v, Y_{W'})$ and $I_1 = I_{\chi^2}(X_u; X_v, Y_{W'}, \tilde Y_{w_0})$, with
	$\tilde Y_{w_0} = X_{u_0} \oplus X_{v_0}$ denoting the noiseless observation. It is then easy to see
	that~\eqref{eq:ab18_st} grows slower (in terms of $(1-2\delta_0)^2$) than the percolation probability.

\section*{Acknowledgement}
	\label{sec:ack}
	Y.~Wu is grateful to Jiaming Xu for discussions pertaining to \prettyref{rmk:generalchannel}.
	Y.~Polyanskiy thanks Emmanuel Abbe for introducing him to and sharing his results on the group synchronization
	over a 2D grid.



\begin{thebibliography}{AMM{\etalchar{+}}18}

\bibitem[AB18a]{AB18-gsync}
E.~Abbe and E.~Boix.
\newblock Broadcasting and synchronizing bits on graphs.
\newblock {\em {\rm Presentation at} Workshop on Combinatorial Statistics}, May
  2018.

\bibitem[AB18b]{abbe2018information}
Emmanuel Abbe and Enric Boix.
\newblock An information-percolation bound for spin synchronization on general
  graphs.
\newblock {\em arXiv preprint arXiv:1806.03227}, 2018.

\bibitem[Abb18]{Abbe2018-april}
E.~Abbe.
\newblock Personal communication, April 2018.

\bibitem[AMM{\etalchar{+}}18]{abbe2017group}
Emmanuel Abbe, Laurent Massouli\'e, Andrea Montanari, Allan Sly, and Nikhil
  Srivastava.
\newblock Group synchronization on grids.
\newblock {\em to appear in Mathematical Statistics and Learning}, 2018.
\newblock arXiv preprint arXiv:1706.08561.

\bibitem[AN72]{AthreyaNey}
Krishna~B Athreya and Peter~E Ney.
\newblock {\em Branching Processes}.
\newblock Springer-Verlag, 1972.

\bibitem[BBAP05]{BBP05}
Jinho Baik, G{\'e}rard Ben~Arous, and Sandrine P{\'e}ch{\'e}.
\newblock Phase transition of the largest eigenvalue for nonnull complex sample
  covariance matrices.
\newblock {\em The Annals of Probability}, 33(5):1643--1697, 2005.

\bibitem[BDM{\etalchar{+}}16]{Barbier16_mutual}
Jean Barbier, Mohamad Dia, Nicolas Macris, Florent Krzakala, Thibault Lesieur,
  and Lenka Zdeborov\'{a}.
\newblock Mutual information for symmetric rank-one matrix estimation: A proof
  of the replica formula.
\newblock In D.~D. Lee, M.~Sugiyama, U.~V. Luxburg, I.~Guyon, and R.~Garnett,
  editors, {\em Advances in Neural Information Processing Systems 29}, pages
  424--432. Curran Associates, Inc., 2016.

\bibitem[BMNN16]{banks2016information}
Jess Banks, Cristopher Moore, Joe Neeman, and Praneeth Netrapalli.
\newblock Information-theoretic thresholds for community detection in sparse
  networks.
\newblock In {\em Conference on Learning Theory}, pages 383--416, 2016.

\bibitem[Csi69]{Csiszar69}
I.~Csisz{\'a}r.
\newblock Eine informationstheoretische ungleichung und ihre anwendung auf den
  beweis der ergodizitat von markoffschen ketten.
\newblock {\em Publ. Math. Inst. Hungar. Acad. Sci., Ser. A}, 8:85--108, 1969.

\bibitem[CT06]{cover}
Thomas~M. Cover and Joy~A. Thomas.
\newblock {\em Elements of information theory, 2nd Ed.}
\newblock Wiley-Interscience, New York, NY, USA, 2006.

\bibitem[DAM16]{deshpande2015asymptotic}
Yash Deshpande, Emmanuel Abbe, and Andrea Montanari.
\newblock Asymptotic mutual information for the two-groups stochastic block
  model.
\newblock {\em Information and Inference: A Journal of the IMA}, 6(2):125--170,
  2016.

\bibitem[EKPS00]{EKPS00}
William Evans, Claire Kenyon, Yuval Peres, and Leonard~J Schulman.
\newblock Broadcasting on trees and the {I}sing model.
\newblock {\em Ann. Appl. Probab.}, 10(2):410--433, 2000.

\bibitem[ES99]{evans1999signal}
William~S Evans and Leonard~J Schulman.
\newblock Signal propagation and noisy circuits.
\newblock {\em {IEEE} Trans. Inf. Theory}, 45(7):2367--2373, 1999.

\bibitem[GP19]{GuPo2019-draft}
Y.~Gu and Y.~Polyanskiy.
\newblock Nonlinear log-{S}obolev inequalities for the {P}otts channel with
  applications to reconstruction problems.
\newblock in preparation, May 2019.

\bibitem[JLR00]{JLR00}
Svante Janson, Tomasz \L{}uczak, and Andrzej Rucinski.
\newblock {\em Random graphs}.
\newblock John Wiley \& Sons, 2000.

\bibitem[Kes80]{Kesten80}
Harry Kesten.
\newblock The critical probability of bond percolation on the square lattice
  equals 1/2.
\newblock {\em Communications in mathematical physics}, 74(1):41--59, 1980.

\bibitem[KM75]{KM75comparison}
J~K\"orner and K~Marton.
\newblock Comparison of two noisy channels.
\newblock {\em Topics in information theory}, pages 411--423, 1975.

\bibitem[KXZ16]{krzakala2016mutual}
Florent Krzakala, Jiaming Xu, and Lenka Zdeborov{\'a}.
\newblock Mutual information in rank-one matrix estimation.
\newblock In {\em 2016 IEEE Information Theory Workshop (ITW)}, pages 71--75.
  IEEE, 2016.

\bibitem[LKZ15]{lesieur2015mmse}
Thibault Lesieur, Florent Krzakala, and Lenka Zdeborov{\'a}.
\newblock Mmse of probabilistic low-rank matrix estimation: Universality with
  respect to the output channel.
\newblock In {\em 2015 53rd Annual Allerton Conference on Communication,
  Control, and Computing}, pages 680--687, 2015.

\bibitem[{\L}uc90]{Luczak90}
Tomasz {\L}uczak.
\newblock Component behavior near the critical point of the random graph
  process.
\newblock {\em Random Structures \& Algorithms}, 1(3):287--310, 1990.

\bibitem[MNS13]{mossel2013proof}
Elchanan Mossel, Joe Neeman, and Allan Sly.
\newblock A proof of the block model threshold conjecture.
\newblock {\em Combinatorica}, pages 1--44, 2013.

\bibitem[MNS15]{mossel2015reconstruction}
Elchanan Mossel, Joe Neeman, and Allan Sly.
\newblock Reconstruction and estimation in the planted partition model.
\newblock {\em Probability Theory and Related Fields}, 162(3-4):431--461, 2015.

\bibitem[MP16]{MP16-comp-arxiv}
A.~Makur and Y.~Polyanskiy.
\newblock Comparison of channels: criteria for domination by a symmetric
  channel.
\newblock {\em arXiv:1609.06877}, September 2016.

\bibitem[PW16]{PW14a}
Yury Polyanskiy and Yihong Wu.
\newblock Dissipation of information in channels with input constraints.
\newblock {\em {IEEE} Trans. Inf. Theory}, 62(1):35--55, January 2016.
\newblock also arXiv:1405.3629.

\bibitem[PW17]{PW15-sdpi-tutorial}
Y.~Polyanskiy and Y.~Wu.
\newblock Strong data-processing inequalities for channels and {B}ayesian
  networks.
\newblock In {\em Convexity, Concentration and Discrete Structures, part of The
  IMA Volumes in Mathematics and its Applications}, volume 161.
  Springer-Verlag, New York, 2017.
\newblock also arXiv:1508.06025.

\bibitem[RU08]{richardson2008modern}
Tom Richardson and Ruediger Urbanke.
\newblock {\em Modern coding theory}.
\newblock Cambridge university press, 2008.

\bibitem[vD97]{vanDijk97}
Marten van Dijk.
\newblock On a special class of broadcast channels with confidential messages.
\newblock {\em IEEE Trans. Inform. Theory}, 43(2):712--714, March 1997.

\end{thebibliography}
\newcommand{\etalchar}[1]{$^{#1}$}

\end{document}